\documentclass[12pt]{article}

  \usepackage{fullpage}
 
\usepackage{amsthm,amsmath,amsfonts}
\usepackage{thm-restate}

\newcommand{\lgr}[1]{\left| #1 \right|}
\newcommand{\zero}{\mathtt{0}}
\newcommand{\one}{\mathtt{1}}
\newcommand{\two}{\mathtt{2}}
\newcommand{\zeon}{\left\{ \one, \two \right\}}
 \newcommand{\zeontw}{\left\{ \zero, \one, \two \right\}}
 \newcommand{\ttx}{\mathtt{x}}
 \newcommand{\mv}{\varepsilon}

\newcommand{\rmt}{\mathrm{t}}

\newcommand{\bQ}{\mathbb{Q}}
\newcommand{\bN}{\mathbb{N}}

\newcommand{\seg}[2]{\left[ #1, #2 \right]}

\newcommand{\mat}[2]{{#2}^{#1 \times #1}}

\newcommand{\pbZ}{\mathfrak{Z}}  
\newcommand{\pbM}{\mathfrak{M}} 
\newcommand{\pbR}{\mathfrak{R}}

\newcommand{\pbgpcp}{\mathrm{GPCP}}
\newcommand{\pbpcp}{\mathrm{PCP}}

\newcommand{\cali}{\mathcal{I}}
\newcommand{\calx}{\mathcal{X}}
\newcommand{\cals}{\mathcal{S}}
\newcommand{\calm}{\mathcal{M}}
 \newcommand{\calk}{\mathcal{K}}
\newcommand{\caly}{\mathcal{Y}}
\newcommand{\call}{\mathcal{L}}

\newcommand{\kp}{k_\mathrm{P}}
\newcommand{\kgp}{k_\mathrm{G}}

\newcommand{\ouru}{\mathbf{U}}
\newcommand{\otheru}{\mathrm{U}}
\newcommand{\otherd}{\mathrm{D}}

\newtheorem{lemma}{Lemma}
\newtheorem{theorem}{Theorem}
\newtheorem{proposition}{Proposition}

\title{Tighter Undecidability Bounds for Matrix Mortality, Zero-in-the-Corner Problems, and More}

\author{Julien Cassaigne \and Vesa Halava \and Tero Harju \and Fran\c{c}ois Nicolas}

\begin{document}

\maketitle
\sloppy

\begin{abstract} 
We study the decidability of three well-known problems related to integer matrix multiplication:
Mortality~($M$),
Zero in the Left-Upper Corner~($Z$),  and
Zero in the Right-Upper Corner~($R$).

Let $d$ and $k$ be positive integers.
Define $M^d(k)$ as the following special case of the Mortality problem:
given a set $\calx$ of $d$-by-$d$ integer matrices such that the cardinality of $\calx$ is not greater than $k$,
decide whether the  $d$-by-$d$ zero matrix belongs to $\calx^+$,
where $\calx^+$ denotes the closure of $\calx$ under the usual matrix multiplication.
In the same way, define the $Z^d(k)$ problem as:
given an instance $\calx$ of $M^d(k)$ (the instances of $Z^d(k)$ are the same as those of $M^d(k)$), 
decide whether at least one matrix in $\calx^+$ has a zero in the left-upper corner.
Define $R^d(k)$ as the variant of $Z^d(k)$ where ``left-upper corner'' is replaced with ``right-upper corner''.
In the paper, 
we prove that 
$M^3(6)$, $M^5(4)$, $M^9(3)$, $M^{15}(2)$,
$Z^3(5)$, $Z^5(3)$, $Z^9(2)$,
$R^3(6)$, $R^4(5)$, and $R^6(3)$ 
are undecidable.
The previous best comparable results were the undecidabilities of $M^3(7)$, $M^{13}(3)$, $M^{21}(2)$,
$Z^3(7)$, $Z^{13}(2)$, $R^3(7)$, and $R^{10}(2)$. 
\end{abstract}

\section{Introduction}

\subsection{Notation and definition}

Given two decision problems $P$ and $P'$, 
we say that $P$ \emph{reduces} to $P'$ if 
there exists an oracle Turing machine \cite{SipserBook} $T$ such that: 
if the oracle solves  $P'$ then $T$ solves $P$.
Two decision problems are called \emph{equivalent} if they reduce to each other. 

As usual, 
$\bN$ denotes the semiring of non-negative integers and 
$\bQ$ denotes the field of rational numbers.
For every  $n \in \bN$,
$\seg{1}{n}$ denotes the set of all $k \in \bN$ such that $1 \le k \le n$.

\subsubsection{Matrices}

For every $m$, $n \in \bN \setminus \{ 0 \}$, 
$\bQ^{m \times n}$ denotes the set of all $m$-by-$n$ matrices with entries in $\bQ$,
$I_n$ denotes the $n$-by-$n$ identity matrix, 
and 
$O_{m, n}$ denotes the $m$-by-$n$ zero matrix; 
subscripts are sometimes dropped when there is no ambiguity.
For every matrix $X$, $X^{\rmt}$ denotes the transpose of $X$.

Let $d \in \bN \setminus \{ 0 \}$.
For every $\calx \subseteq \bQ^{d \times d}$,
define $\calx^+$ as the closure of $\calx$ under the usual matrix multiplication and 
define $ \calx^\star = \calx^ + \cup \{ I_d \}$.
For every $X \in \mat{d}{\bQ}$, 
$X^\star$ is understood as a shorthand for $\left\{ X \right\}^\star = \left\{ X^n : n \in \bN \right\}$.

\subsubsection{Semigroups}

A \emph{semigroup} is a set equipped with an associative operation.
A \emph{monoid} is a semigroup that has an identity element.
For instance, $\mat{d}{\bQ}$ is a monoid under the usual matrix multiplication.
For every $\calx \subseteq \mat{d}{\bQ}$, 
$\calx^+$ is the multiplicative subsemigroup of $\mat{d}{\bQ}$ generated by $\calx$
and 
$\calx^\star$ is the multiplicative submonoid of $ \mat{d}{\bQ}$ generated by $\calx$. 

Let $\cals$ and $\cals'$ be multiplicative semigroups.
A  \emph{morphism} from $\cals$ to $\cals'$ is a function $\Phi\colon \cals \to \cals'$ 
such that $\Phi(XY) = \Phi(X) \Phi(Y)$ for all $X$, $Y \in \cals$.
Throughout the paper, 
``morphism'' always means ``multiplicative semigroup morphism''.


\subsection{Problems}

 Let $d \in \bN \setminus \{ 0 \}$.

The \emph{Zero Reachability problem over $\mat{d}{\bQ}$} \cite{ HalavaH07, HalavaHH07, GaubertK06, SalomaaS78}, 
denoted $\pbZ^d$, 
is defined as:
given $L \in \bQ^{1 \times d}$, $C \in \bQ^{d \times 1}$, 
and 
a finite $\calx \subseteq \mat{d}{\bQ}$, 
decide whether there exists $Y \in \calx^+$ such that $L Y C = 0$.

For every $i$, $j \in \seg{1}{d}$, 
the following problem is denoted $\pbZ_{i, j}^d$:
given a finite $\calx \subseteq \mat{d}{\bQ}$,
decide whether there exists $Y \in \calx^+$ such that the $(i, j)$th entry of $Y$ equals~$0$.
$\pbZ_{1, 1}^d$ is the \emph{Zero in the Left-Upper Corner problem over $\mat{d}{\bQ}$}
\cite{BournezB02,  HalavaH07, HalavaHH07, HalavaH01}.
Put $\pbR^d = \pbZ_{1, d}^d$.
$\pbR^d$ is the \emph{Zero in the Right-Upper Corner problem over $\mat{d}{\bQ}$}
\cite{Manna74, GaubertK06, CassaigneK98, HarjuKHandbook, BellP08bounds, HalavaH07, HalavaHH07, Claus80, BellP12complexity}.

The \emph{Mortality problem over  $\mat{d}{\bQ}$} \cite{BournezB02, HalavaHH07, HalavaH01, Paterson70, CassaigneK98, BlondelT97, KromKrom90, Schultz77, Miller94, HarjuKHandbook, BellHP12},
denoted $\pbM^d$, 
is defined as: 
given a finite $\calx \subseteq \mat{d}{\bQ}$,
decide whether the $d$-by-$d$ zero matrix belongs to $\calx^+$.

Let $k \in \bN$.
Define $\pbZ^d(k)$ as the restriction of $\pbZ^d$ to those instances $(L, C, \calx)$ for which the cardinality of $\calx$ is not greater than $k$.
For every $i$, $j \in \seg{1}{d}$,  
define $\pbZ_{i, j}^d(k)$ as the restriction of $\pbZ^d_{i, j}$ to those subsets of $\mat{d}{\bQ}$ that have cardinality $k$ or less.
Put $\pbR^d(k) = \pbZ_{1, d}^d(k)$.
Define $\pbM^d(k)$ as the restriction of $\pbM^d$ to those subsets of $\mat{d}{\bQ}$ that have cardinality $k$ or less.
We convene that 
$\pbZ^d(\infty) = \pbZ^d$, 
$\pbZ_{i, j}^d(\infty) = \pbZ^d_{i, j}$, 
$\pbR^d(\infty) = \pbR^d$, and 
$\pbM^d(\infty) = \pbM^d$.

Note that restricting the previously defined problems to matrices with integer entries does not modify their decidabilities. 
Restricting them to matrices with non-negative integer entries makes them decidable \cite{BlondelT97, GaubertK06}.

\subsection{Organization of the paper}

The paper is divided into five sections.
Let $d \in \bN \setminus \{ 0 \}$ and let $k \in \bN \cup \{ \infty \}$.
In Section~\ref{sec:general}, we prove the following four propositions:

\begin{restatable}{proposition}{ZijR} \label{prop:ZijR}
For every $i$, $j \in \seg{1}{d}$ with $i \ne j$,  $\pbZ_{i, j}^d(k)$ is equivalent to~$\pbR^d(k)$.
 \end{restatable}

 \begin{restatable}{proposition}{ZiiZ} \label{prop:ZiiZ}
For every $i \in \seg{1}{d}$,
 $\pbZ_{i, i}^d(k)$ is equivalent to~$\pbZ^d(k)$.   
 \end{restatable}
 
 \begin{restatable}{proposition}{ZR} \label{prop:ZR}
 $\pbZ^d(k)$ reduces to $\pbR^{d + 1}(k)$.
 \end{restatable}

  \begin{restatable}{proposition}{ZM} \label{prop:ZM}
 $\pbZ^d(k)$  reduces to $\pbM^d(k + 1)$.
  \end{restatable}
Note that the equivalence of $\pbZ^d(k)$ and $\pbZ_{1, 1}^d(k)$, which follows from Proposition~\ref{prop:ZiiZ}, was previously overlooked.
In Section~\ref{sec:three-three}, 
we prove that  $\pbZ^3(5)$ and $\pbR^3(6)$ are undecidable.
In Section~\ref{sec:high-dim}, 
we prove that   $\pbZ^5(3)$,  $\pbZ^9(2)$, and  $\pbM^{15}(2)$, are undecidable.
In Section~\ref{sec:open}, we put forward some remaining open questions.

\subsection{Contribution}

The undecidabilities of $\pbZ^3(5)$, $\pbZ^5(3)$, and $\pbZ^9(2)$ imply those of 
$\pbR^4(5)$,
$\pbR^6(3)$, and 
$\pbR^{10}(2)$
by Proposition~\ref{prop:ZR}
and 
those of 
$\pbM^3(6)$,
$\pbM^5(4)$, and 
$\pbM^9(3)$
by Proposition~\ref{prop:ZM}.
Hence, the following problems are proven undecidable in the present paper:
$\pbZ^3(5)$,
$\pbZ^5(3)$, 
$\pbZ^9(2)$,
$\pbR^3(6)$,
$\pbR^4(5)$,
$\pbR^6(3)$,
$\pbR^{10}(2)$,
$\pbM^3(6)$,
$\pbM^5(4)$, 
$\pbM^9(3)$, and
$\pbM^{15}(2)$;
the undecidabilities of $\pbZ^3(5)$, $\pbZ^9(2)$, and $\pbR^{10}(2)$ were previously known \cite{HalavaHH07,HalavaH07}.
Previous results about our problems are summarized in Table~\ref{tab:prev-work}.
Our contribution is depicted in Tables~\ref{tab:ZR}, \ref{tab:ZLUC}, \ref{tab:ZRUC}, and~\ref{tab:Mort}.
The contents of the five tables are to be understood as follows:
 $\otherd$ stands for ``decidable'',
$\otheru$ and $\ouru$ stand for ``undecidable'', 
$?$ stands for ``unknown'', and 
$\ouru$ denotes our contribution.

\begin{table}
\begin{center}
 \begin{tabular}{|l|l|l|}
\hline
Problem & Status & Reference(s)  \\
\hline
$\pbZ^1(5)$ & $\otherd$  & \cite{HHHK-Skolem-05} \\
$\pbZ^3(5)$ & $\otheru$ & \cite{HalavaHH07}  \\
$\pbZ^9(2)$ & $\otheru$ & \cite{HalavaH07}  \\
\hline
$\pbZ_{1, 1}^3(7)$    & $\otheru$ &\cite{HalavaHH07} \\
$\pbZ_{1, 1}^{13}(2)$ & $\otheru$ & \cite{HalavaH07} \\
\hline
$\pbR^3(7)$    & $\otheru$ & \cite{Claus80,HarjuKHandbook,MatiyasevichS05}  \\
$\pbR^{10}(2)$ & $\otheru$ & \cite{HalavaH07} \\
\hline
$\pbM^2(2)$    & $\otherd$ & \cite{BournezB02} \\
$\pbM^2$       & NP-hard & \cite{BellHP12} \\
$\pbM^3(7)$    & $\otheru$ & \cite{HalavaHH07} \\
$\pbM^{13}(3)$ & $\otheru$ & \cite{HalavaH07} \\
$\pbM^{21}(2)$ & $\otheru$ & \cite{HalavaHH07} \\
\hline
\end{tabular}
\end{center}
\caption{\label{tab:prev-work} Previous work.}
\end{table}

\begin{table} 
$$
\begin{array}{|cc|cc ccc c}
\hline
& & k \\ 
    & &  2 & 3 & 4 & 5 & 6 & \cdots \\
\hline 
d & 2 &  ? & ? & ? & ? & ? & \cdots \\
  & 3 & ? & ? & ? & \otheru & \otheru & \cdots \\
  & 4 & ? & ? & ? & \otheru & \otheru & \cdots \\
  & 5 &  ? & \ouru & \ouru & \otheru & \otheru & \cdots \\
  & 6 &  ? & \ouru & \ouru & \otheru & \otheru & \cdots \\
  & 7 &  ? & \ouru & \ouru & \otheru & \otheru & \cdots \\
  & 8 &  ? & \ouru & \ouru & \otheru & \otheru & \cdots \\
  & 9 & \otheru & \otheru & \otheru & \otheru & \otheru & \cdots \\
 & 10 &  \otheru & \otheru & \otheru & \otheru & \otheru & \cdots \\
 & \vdots &  \vdots &  \vdots &  \vdots &  \vdots &   \vdots & \ddots 
\end{array}
$$
\caption{\label{tab:ZR} Current knowledge about the undecidability of $\pbZ^d(k)$.}
\end{table}

\begin{table}
$$
\begin{array}{|cc|ccccccccc}
\hline
& & k \\ 
    & &  2 & 3 & 4 & 5 & 6 & 7 & 8 & \cdots \\
\hline 
d & 2  & ? & ? & ? & ? & ? & ? & ? &  \cdots \\
  & 3  & ? &     ? &     ? & \ouru & \ouru & \otheru & \otheru & \cdots \\
  & 4  & ? &     ? &     ? & \ouru & \ouru & \otheru & \otheru & \cdots \\
  & 5  & ? & \ouru & \ouru & \ouru & \ouru & \otheru & \otheru & \cdots \\
  & 6  & ? & \ouru & \ouru & \ouru & \ouru & \otheru & \otheru & \cdots \\
  & 7  & ? & \ouru & \ouru & \ouru & \ouru & \otheru & \otheru & \cdots \\
  & 8  & ? & \ouru & \ouru & \ouru & \ouru & \otheru & \otheru & \cdots \\
  & 9  & \ouru  & \ouru & \ouru & \ouru & \ouru & \otheru & \otheru &\cdots \\
  & 10 & \ouru  & \ouru & \ouru & \ouru & \ouru & \otheru & \otheru &\cdots \\
  & 11 & \ouru  & \ouru & \ouru & \ouru & \ouru & \otheru & \otheru &\cdots \\
  & 12 & \ouru  & \ouru & \ouru & \ouru & \ouru & \otheru & \otheru &\cdots \\
  & 13 & \otheru  & \otheru & \otheru & \otheru & \otheru & \otheru & \otheru &\cdots \\
  & 14 & \otheru  & \otheru & \otheru & \otheru & \otheru & \otheru & \otheru &\cdots \\
  & \vdots &  \vdots &  \vdots &  \vdots &  \vdots &   \vdots &  \vdots & \vdots &  \ddots 
\end{array}
$$
\caption{\label{tab:ZLUC} Current knowledge about the undecidability of $\pbZ_{1, 1}^d(k)$.}
\end{table}

 \begin{table}
 $$
 \begin{array}{|cc|cc ccc ccc}
 \hline
 & & k \\ 
     & &  2 & 3 & 4 & 5 & 6 & 7 & 8 & \cdots \\
 \hline 
 d & 2 &  ? &     ? &     ? &   ? & ? & ? & ? & \cdots \\
   & 3 &  ? &     ? &     ? & ? & \ouru & \otheru & \otheru & \cdots \\
   & 4 &  ? &     ? &     ? & \ouru & \ouru & \otheru & \otheru & \cdots \\
   & 5 &  ? &     ? &     ? & \ouru & \ouru & \otheru & \otheru & \cdots \\
   & 6 &  ? & \ouru & \ouru & \ouru & \ouru & \otheru & \otheru & \cdots \\
   & 7 &  ? & \ouru & \ouru & \ouru & \ouru & \otheru & \otheru & \cdots \\
   & 8 &  ? & \ouru & \ouru & \ouru & \ouru & \otheru & \otheru & \cdots \\
   & 9 &  ?  & \ouru & \ouru & \ouru & \ouru & \otheru & \otheru &\cdots \\
   & 10 & \otheru  & \otheru & \otheru & \otheru & \otheru & \otheru & \otheru &\cdots \\
   & 11 & \otheru  & \otheru & \otheru & \otheru & \otheru & \otheru & \otheru &\cdots \\
   & \vdots&  \vdots &  \vdots &  \vdots &  \vdots &   \vdots & \vdots &   \vdots & \ddots 
 \end{array}
 $$
  \caption{\label{tab:ZRUC} Current knowledge about the undecidability of $\pbR^d(k)$.}
\end{table}

 \begin{table}
 $$
 \begin{array}{|cc|cc ccc ccc}
 \hline
 & & k \\ 
     & &  2 & 3 & 4 & 5 & 6 & 7 & 8 & \cdots \\
 \hline 
 d &  2 & \otherd & ?       & ?       & ?       & ? & ? & ? & \cdots \\
   &  3 & ?       & ?       & ?       & ?       & \ouru   & \otheru & \otheru & \cdots \\
   &  4 & ?       & ?       & ?       & ?       & \ouru   & \otheru & \otheru & \cdots \\
   &  5 & ?       & ?       & \ouru   & \ouru   & \ouru   & \otheru & \otheru & \cdots \\
   &  6 & ?       & ?       & \ouru   & \ouru   & \ouru   & \otheru & \otheru & \cdots \\
   &  7 & ?       & ?       & \ouru   & \ouru   & \ouru   & \otheru & \otheru & \cdots \\
   &  8 & ?       & ?       & \ouru   & \ouru   & \ouru   & \otheru & \otheru & \cdots \\
   &  9 & ?       & \ouru   & \ouru   & \ouru   & \ouru   & \otheru & \otheru &\cdots \\
   & 10 & ?       & \ouru   & \ouru   & \ouru   & \ouru   & \otheru & \otheru &\cdots \\
   & 11 & ?       & \ouru   & \ouru   & \ouru   & \ouru   & \otheru & \otheru &\cdots \\
   & 12 & ?       & \ouru   & \ouru   & \ouru   & \ouru   & \otheru & \otheru &\cdots \\
   & 13 & ?       & \otheru & \otheru & \otheru & \otheru & \otheru & \otheru &\cdots \\
   & 14 & ?       & \otheru & \otheru & \otheru & \otheru & \otheru & \otheru &\cdots \\
   & 15 & \ouru   & \otheru & \otheru & \otheru & \otheru & \otheru & \otheru &\cdots \\
   & 16 & \ouru   & \otheru & \otheru & \otheru & \otheru & \otheru & \otheru &\cdots \\
   & 17 & \ouru   & \otheru & \otheru & \otheru & \otheru & \otheru & \otheru &\cdots \\
   & 18 & \ouru   & \otheru & \otheru & \otheru & \otheru & \otheru & \otheru &\cdots \\
   & 19 & \ouru   & \otheru & \otheru & \otheru & \otheru & \otheru & \otheru &\cdots \\
   & 20 & \ouru   & \otheru & \otheru & \otheru & \otheru & \otheru & \otheru &\cdots \\
   & 21 & \otheru & \otheru & \otheru & \otheru & \otheru & \otheru & \otheru &\cdots \\
   & 22 & \otheru & \otheru & \otheru & \otheru & \otheru & \otheru & \otheru &\cdots \\
   & \vdots &  \vdots &  \vdots &  \vdots &  \vdots &   \vdots & \vdots &   \vdots & \ddots 
 \end{array}
 $$
  \caption{\label{tab:Mort} Current knowledge about the decidability of $\pbM^d(k)$.}
\end{table}

\section{General results} \label{sec:general}

In this section, 
we prove some basic properties of our problems.
Unsurprisingly, we shall see that they are closely related to each other.
Let $d \in \bN \setminus \{ 0 \}$ and let $k \in \bN \cup \{ \infty \}$.

Set 
$$
E_i = 
 \begin{pmatrix}
  O_{i - 1,1} \\ 1 \\ O_{d - i,1} 
 \end{pmatrix}
$$
for every $i \in \seg{1}{d}$:
the  $d$-tuple $\left( E_i \right)_{i \in \seg{1}{d}}$ is the canonical basis of the linear space $\bQ^{d \times 1}$.
Remark that, for 
any  $L \in  {\bQ}^{1 \times d}$, 
any $C \in  {\bQ}^{d \times 1}$,
any $Y \in {\bQ}^{d \times d}$,
and 
any $i$, $j \in \seg{1}{d}$, 
$L E_j$ equals the $j$th entry of $L$,
$E_i^\rmt C$ equals the $i$th entry of $C$,
$E_i^ \rmt Y$ equals the $i$th row of $Y$,
$Y E_j$ equals the $j$th column of $Y$, and
$E_i^\rmt Y E_j$ equals the $(i, j)$th entry of $Y$.

\begin{lemma} \label{lem:mat-perm}
Let $D \subseteq \seg{1}{d}$ and let $\pi\colon D \to \seg{1}{d}$ be injective.
There exists $P \in \mat{d}{\bQ}$ such that $P$ is non-singular and $P E_j = E_{\pi(j)}$ for every $j \in D$. 
\end{lemma}

\begin{proof}
Let us first consider the case where $D = \seg{1}{d}$,
\emph{i.e}, where $\pi$ is a permutation of $\seg{1}{d}$.
Let $P$ be the permutation matrix associated with $\pi$:
$$
P = \sum_{i = 1}^ d E_{ \pi (i)} E_i^ \rmt \, .
$$
It is easy to see that $P$ satisfies the desired properties; in particular, 
note that  $P^\rmt = P^{-1}$ is the permutation matrix associated with $\pi^{-1}$.

Let us now deal with the general case.
Remark that there exists a permutation $\bar \pi$ of $\seg{1}{d}$ such that $\bar \pi(j) = \pi(j)$ for every $j \in D$.
Hence, the general case reduces to the case where $D = \seg{1}{d}$.
\end{proof}

\begin{lemma} \label{lem:ZiiZ11}
For every $i \in \seg{1}{d}$, 
 $\pbZ_{i, i}^d(k)$ is equivalent to $\pbZ_{1, 1}^d(k)$.
\end{lemma}

\begin{proof}
Let $i$, $j \in \seg{1}{d}$ be fixed.
Applying Lemma~\ref{lem:mat-perm} with $D = \{ i, j \}$, 
we see that there exists $P \in\mat{d}{\bQ}$ such that $P$ is non-singular, $P E_{i} = E_{j}$, and $P E_j = E_{i}$.
Let $\Phi\colon \mat{d}{\bQ} \to \mat{d}{\bQ}$ be the morphism defined by:
$\Phi(X) = P X P^{-1}$ for every $X \in \mat{d}{\bQ}$.
For every $Y \in \mat{d}{\bQ}$,
the $(i, i)$th entry of  $Y$ equals the $(j, j)$th entry of $\Phi(Y)$.
Hence, 
$\Phi$ induces a reduction from $\pbZ^d_{i, i}(k)$ to $\pbZ^d_{j, j}(k)$.
\end{proof}


\ZijR*

\begin{proof}
Let $i_1$, $j_1$, $i_2$, $j_2 \in \seg{1}{d}$ be such that $i_1 \ne j_1$ and $i_2 \ne j_2$.
Applying Lemma~\ref{lem:mat-perm} with $D = \{ i_1, j_1 \}$, 
we see that there exists $P \in\mat{d}{\bQ}$ such that $P$ is non-singular, $P E_{i_1} = E_{i_2}$, and $P E_{j_1} = E_{j_2}$.
Let $\Phi\colon \mat{d}{\bQ} \to \mat{d}{\bQ}$ be the morphism defined by:
$\Phi(X) = P X P^{-1}$ for every $X \in \mat{d}{\bQ}$.
For every $Y \in \mat{d}{\bQ}$,
the $(i_1, j_1)$th entry of  $Y$ equals the $(i_2, j_2)$th entry of $\Phi(Y)$.
It follows that $\Phi$ induces a reduction from $\pbZ_{i_1, j_1}^d(k)$ to $\pbZ^d_{i_2, j_2}(k)$.
\end{proof}

\begin{lemma} \label{lem:ZijZ}
 For every $i$, $j \in \seg{1}{d}$, $\pbZ^d_{i, j}(k)$ reduces to $\pbZ^d(k)$.
\end{lemma}

\begin{proof}
For every finite $\calx \subseteq \mat{d}{\bQ}$, 
$\calx$ is a yes-instance of $\pbZ^d_{i, j}$
if, and only if, 
$(E_i^\rmt, E_j, \calx)$ is a yes-instance of $\pbZ^d$.
\end{proof}

An instance $(L, C, \calx)$ of $\pbZ^d$ is called \emph{non-degenerated} if $L C \ne 0$.

\begin{lemma}[\cite{GaubertK06}] \label{lem:S-Star}
 $\pbZ^d(k)$ reduces to its restriction to non-degenerated instances.
\end{lemma}

\begin{proof}
Let $(L, C, \calx)$ be an instance of $\pbZ^d(k)$.

First, assume that $L X C = 0$ for some $X \in \calx$.
Then,  $(L, C, \calx)$ is a yes-instance of $\pbZ^d$.

Second, assume that $L X C \ne 0$ for every $X \in \calx$.
Then, $(L, X C, \calx)$ is a non-degenerated instance of $\pbZ^d(k)$ for every $X \in \calx$.
Moreover,  $(L, C, \calx)$ is a yes-instance of $\pbZ^d$ 
if, and only if,
there exists $X \in \calx$ such that $(L, X C, \calx)$ is a yes-instance of~$\pbZ^d$.
\end{proof}

\begin{lemma} \label{lem:LPEC}
Let $L \in \bQ^{1 \times d}$ and let $C \in \bQ^{d \times 1}$ be such that $LC \ne 0$.
There exists $P \in \mat{d}{\bQ}$ such that $P$ is non-singular,  $L P = LC E_1^{\rmt}$, and  $C =  P E_1$.
\end{lemma}

\begin{proof}
First, consider the case where both the leftmost entry of $L$ and the uppermost entry of $C$ equal $1$.
Then, there exist  $L' \in {\bQ}^{1 \times (d - 1)}$ and  $C' \in {\bQ}^{(d - 1) \times 1}$ such that 
\begin{align*}
 L &  = \begin{pmatrix}  
 1 & L' 
\end{pmatrix}
& & \text{and} & 
C & = \begin{pmatrix}
 1 \\ C' 
\end{pmatrix} \, .
\end{align*}
Put  
\begin{align*}
U &  = \begin{pmatrix}  
 1 & O \\
 C' & -I
\end{pmatrix} \, , &
V & = \begin{pmatrix}
 1 & -  {(LC)}^{-1} L'  \\
 O & - I
\end{pmatrix} \,, & 
& \text{and} & 
P = UV \,.
\end{align*}
It is clear that $U$, $V$, and $P$ are non-singular with $U^{-1} = U$, $V^{-1} = V$ and $P^{-1} = VU$.
Moreover, we have $L U = \begin{pmatrix} LC &  -L' \end{pmatrix}$ and $ V E_1 = E_1$,
so $L P = LC E_1^{\rmt}$ and  $C =  P E_1$. 

Let us now deal with the general case.
For each $i \in \seg{1}{d}$, put $\lambda_i = L E_i$ and $\gamma_i = E_i^\rmt C$.
Since 
$$
 \sum_{i = 1}^d \lambda_i \gamma_i  = L C \ne 0 \, ,
$$
there exists $j \in \seg{1}{d}$ such that $\lambda_j \gamma_j \ne 0$.
Applying Lemma~\ref{lem:mat-perm} with $D = \{ 1 \}$, 
we see that there exists $T \in\mat{d}{\bQ}$ such that $T$ is non-singular and $T E_1 = E_j$.
Put 
$\bar L  = \lambda_j^{-1} L T$
and
$\bar C  =    \gamma_j^{-1} T^{-1} C$.
By construction, 
we have $\bar L \bar C = \lambda_j^{-1}  \gamma_j^{-1}  LC\ne 0$ 
and 
both the leftmost entry of $\bar L$ and the uppermost entry of $\bar C$ equal $1$.
Therefore, there exists $\bar P \in \mat{d}{\bQ}$ such that 
$\bar P$ is non-singular,
$\bar L \bar P = \bar L \bar C E_1^{\rmt}$, and
$\bar C =  \bar P E_1$.
Put $P = \gamma_j T \bar P$.
It is easy to see that $P$ satisfies the desired properties.
\end{proof}

Let $R$ be a division ring, 
let $L$, $L' \in R^{1 \times d}$, and 
let $C$, $C' \in R^{d \times 1}$ be such that none of $C$, $L$, $C'$, and $L'$ is a zero matrix.
We claim that $L C = L'C'$ if, and only if, there exists 
$P \in \mat{d}{R}$ such that $P$ is multiplicatively invertible in $\mat{d}{R}$,  $L P = L'$, and  $C =  PC'$.
Our claim nicely generalizes Lemma~\ref{lem:LPEC}; 
its proof is left to the reader.
It follows from our claim that the restriction of $\pbZ^d(k)$ to degenerated instances is equivalent to~$\pbR^d(k)$; 
the verification is left to the reader.

\ZiiZ*

\begin{proof} 
By Lemmas~\ref{lem:ZiiZ11} and \ref{lem:ZijZ}, 
it suffices to show that  $\pbZ^d(k)$ reduces to $\pbZ_{1, 1}^d(k)$.
Moreover, 
by Lemma~\ref{lem:S-Star}, we only need to reduce non-degenerated instances of~$\pbZ^d(k)$.

Let $(L, C, \calx)$ be a non-degenerated instance of $\pbZ^d(k)$.
By Lemma~\ref{lem:LPEC}, there exists  $P \in \mat{d}{\bQ}$ such that $P$ is non-singular,  
$L P = LC E_1^{\rmt}$, and  $C =  P E_1$.
Put $\calx' = \left\{ P^{-1} X P  : X \in \calx \right\}$.
Since the cardinality of $\calx'$ equals that of $\calx$,  
$\calx'$ is an instance of $\pbZ_{1, 1}^d(k)$. 
Moreover, 
$P$ is computable from $L$ and $C$ (a more efficient method than brute-force enumeration can be derived from a simple examination of the proof of Lemma~\ref{lem:LPEC}), 
so $\calx'$ is computable from $(L, C, \calx)$.
Finally, remark that for every $Y \in \mat{d}{\bQ}$, 
the $(1, 1)$th entry of $ P^{-1} Y P$  equals ${(LC)}^{-1} L Y C$. 
Therefore, $(L, C, \calx)$ is a yes-instance of $\pbZ^d$ if, and only if, $\calx'$ is a yes-instance of $\pbZ_{1, 1}^d$.
\end{proof}

Lemma~\ref{lem:ZijZ} ensures that $\pbR^d(k)$ reduces to $\pbZ^d(k)$; 
whether $\pbZ^d(k)$  reduces to $\pbR^d(k)$ is an open question.
However, it holds true that:

\ZR*

\begin{proof}
By Proposition~\ref{prop:ZiiZ}, it suffices to prove that $\pbZ_{1, 1}^d(k)$  reduces to  $\pbR^{d + 1}(k)$.

Let $\Phi\colon \mat{d}{\bQ} \to \mat{(d + 1)}{\bQ}$ be the morphism defined by:
$$
\Phi(X) = 
\begin{pmatrix}
X & X E_1 \\
O  & 0
\end{pmatrix}
$$
for every $X \in \mat{d}{\bQ}$.
For every $Y \in \mat{d}{\bQ}$,
the $(1, 1)$th entry of $Y$ equals the $(1, d + 1)$th entry of $\Phi(Y)$.
Hence, $\Phi$ induces a reduction from $\pbZ_{1, 1}^d(k)$  to  $\pbR^{d + 1}(k)$.
\end{proof}

Proposition~\ref{prop:ZR} improves on the following result, 
which is implicitly used in at least two papers:

\begin{proposition}[\cite{HalavaH07,GaubertK06}] \label{prop:S-K+2}
 $\pbZ^d(k)$ reduces to $\pbR^{d + 2}(k)$.
\end{proposition}

\begin{proof}
For every $L \in \bQ^{1 \times d}$ and every $C \in \bQ^{d \times 1}$,
let $\Phi_{L, C} \colon \mat{d}{\bQ} \to \mat{(d + 2)}{\bQ}$ be the  morphism defined by:
$$
\Phi_{L, C}(X) = 
\begin{pmatrix}
 0 & L X & LXC \\
 O & X & X C \\
 0 & O & 0
\end{pmatrix}
$$
for every $X \in \mat{d}{\bQ}$.
For every instance $(L, C, \calx)$ of $\pbZ^d$,
$(L, C, \calx)$ is a yes-instance of  $\pbZ^d$ 
if, and only if,
$\Phi_{L, C}(\calx)$ is a yes-instance of $\pbR^{d + 2}(k)$.
\end{proof}

\begin{proposition} \mbox{}
\begin{itemize}
  \item $\pbR^d(k)$ reduces to $\pbR^{d + 1}(k)$.
  \item $\pbZ^d(k)$ reduces to $\pbZ^{d + 1}(k)$.
  \item $\pbM^d(k)$ reduces to $\pbM^{d + 1}(k)$.
\end{itemize}
 \end{proposition}
 
 \begin{proof}
Remark that $\pbR^d(k)$ reduces to $\pbZ^d(k)$ by Lemma~\ref{lem:ZijZ}
and that 
$\pbZ^d(k)$ reduces to $\pbR^{d + 1}(k)$ by Proposition~\ref{prop:ZR}.
Therefore,  the first part of the proposition holds true.

The second part  of the proposition can be proven in the same way:
$\pbZ^d(k)$ reduces to $\pbR^{d + 1}(k)$ by Proposition~\ref{prop:ZR}
and 
$\pbR^{d + 1}(k)$ reduces to $\pbZ^{d + 1}(k)$ by Lemma~\ref{lem:ZijZ}.

Let $\Phi \colon \mat{d}{\bQ} \to \mat{(d + 1)}{\bQ}$ be the  morphism defined by:
$$
\Phi(X) = 
\begin{pmatrix}
X & O \\
O & 0
\end{pmatrix}
$$
for every $X \in \mat{d}{\bQ}$.
For every $Y \in \mat{d}{\bQ}$,
$Y$  equals the $d$-by-$d$ zero matrix 
if, and only if, 
$\Phi(Y)$ equals $(d + 1)$-by-$(d + 1)$ zero matrix.
Hence, $\Phi$ induces a reduction from $\pbM^d(k)$  to $\pbM^{d + 1}(k)$,
and thus the third part of the proposition holds true.
 \end{proof}

%

\begin{lemma} \label{lem:Zd-Md}
 Let $L \in \bQ^{1 \times d}$, let $C \in \bQ^{d \times 1}$, and let $\calx \subseteq \mat{d}{\bQ}$.
 The following two assertions are equivalent:
 \begin{enumerate}
  \item There exists $Y \in  \calx^\star$ such that  $L Y C = 0$.
 \item The $d$-by-$d$ zero matrix belongs to $\left( \calx \cup \{ CL \} \right)^+$.
 \end{enumerate}
  
\end{lemma}

\begin{proof}
Note that $L \calx^\star C \subseteq \bQ$.
Since $(L Y C) CL  \in \left(  \calx \cup \{ CL \}  \right)^+$ for every $Y \in \calx^ \star$,
the first considered assertion implies the second one.
Since
$$
L \left( \calx \cup \{ CL \} \right)^\star C  = \left( L \calx^\star C  \right)^+ \,, 
$$
the second considered assertion implies $0 \in \left( L \calx^\star C  \right)^+$.
Besides, 
$ 0 \in \left( L \calx^\star C  \right)^+$ is equivalent to $0 \in L \calx^\star C$
because $\bQ$ has the zero-product property.
Therefore, the considered assertions are equivalent.
\end{proof}

\ZM*

\begin{proof}
By Lemma~\ref{lem:S-Star}, we only need to reduce non-degenerated instances of $\pbZ^d(k)$.

Let $(L, C, \calx)$ be a non-degenerated instance of  $\pbZ^d(k)$.
Clearly, $\calx \cup \{ CL \}$ is an instance of $\pbM^d(k + 1)$ and $\calx \cup \{ CL \}$ is computable from $(L, C, \calx)$.
To conclude the proof of the proposition, we only need to check that the following three assertions are equivalent:
 \begin{enumerate}
  \item $(L, C, \calx)$ is a yes-instance  of  $\pbZ^d$.
  \item There exists $Y \in  \calx^\star$ such that  $L Y C = 0$.
  \item $\calx \cup \{ CL \}$ is a yes-instance of $\pbM^d$.
 \end{enumerate}
The first two considered assertions are equivalent because $LC \ne 0$;
the last two considered assertions are equivalent by Lemma~\ref{lem:Zd-Md}.
\end{proof}

\begin{proposition}[\cite{BournezB02}] \label{prop:Z2-M2}
 $\pbZ^2(k)$  is equivalent to $\pbM^2(k + 1)$.
\end{proposition}

\begin{proof}
By Proposition~\ref{prop:ZM}, 
it suffices to show that  $\pbM^2(k + 1)$ reduces to $\pbZ^2(k)$.
The proof is based on Lemma~\ref{lem:Zd-Md} and the following property of $2$-by-$2$ matrices: 
for every $X \in \mat{2}{\bQ}$, 
either $X$ is non-singular 
or $X$ can be written as an outer product.

Let $\calx$ be an instance of  $\pbM^2(k + 1)$.

First, assume that all matrices in $\calx$ are non-singular.
Then, $\calx$ is a no-instance of $\pbM^2$ because  all matrices in $\calx^+$ are non-singular.

Second, assume that some matrix in $\calx$ can be written as an outer product.
Then, there exist $L \in {\bQ}^{1 \times 2}$ and $C \in \bQ^{2 \times 1}$ such that $CL \in \calx$. 
Clearly, $(L, C, \calx \setminus \{ C L \})$ is an instance of $\pbZ^2(k)$ and $(L, C, \calx \setminus \{ C L \})$ is computable from $\calx$.
To conclude the proof of the proposition, we only need to check that the following three assertions are equivalent:
\begin{enumerate}
 \item $\calx$ is a yes-instance of $\pbM^2$.
 \item There exists $Y \in \left( \calx \setminus \{ C L \} \right)^\star$ such that $L Y C = 0$.
 \item $LC = 0$ or $(L, C, \calx \setminus \{ C L \})$ is a yes-instance of $\pbZ^2$.
\end{enumerate}
The first two considered assertions are equivalent by Lemma~\ref{lem:Zd-Md};
the last two considered assertions are clearly equivalent.
\end{proof}

\section{Three-by-three matrices} \label{sec:three-three}

In this section, 
we prove that $\pbZ^3(5)$ and $\pbR^3(6)$  are undecidable by reduction from the generalized Post correspondence problem.
Let  $k \in \bN \cup \{ \infty \}$.

\subsection{The (generalized) Post correspondence problem}

Precise definitions of the  \emph{Post Correspondence Problem} (PCP) \cite{Post46PCP} 
and 
its best-known generalization are presented in this section.

An \emph{alphabet} is a finite set of symbols.
The canonical alphabet is the binary alphabet $\zeon$.
A \emph{word} is a finite sequence of symbols.
Word concatenation is denoted  multiplicatively.
For every word $w$, $\lgr{w}$ denotes the \emph{length} of $w$.
The word of length $0$ is called the \emph{empty word} and denoted $\mv$.
Let $A$ be an alphabet.
The set of all words over $A$ is denoted $A^\star$.
Note that $A^\star$ is a monoid under concatenation.
Set $A^+ = A^\star \setminus \{ \mv \}$.

Two slightly different definitions of the \emph{Generalized Post Correspondence Problem} (GPCP) can be found in the literature.
Let $e \in \{ {\star}, {+} \}$. 
Define $\pbgpcp_e$ as the following problem:
given an alphabet $A$, 
two morphisms $f$, $g\colon A^\star \to \zeon^\star$,
and 
$x$, $x'$, $y$, $y' \in \zeon^\star$, 
decide whether there exists $w \in A^e$ such that $x f(w) x' = y g(w) y'$;
it is understood that the instance $(A, f, g, x, x', y, y')$ is encoded by the quintuple 
$$(\left\{ (f(a), g(a)) : a \in A \right\}, x, x', y, y') \, .$$
Define $\pbgpcp_e(k)$ as the restriction of $\pbgpcp_e$ to those instances $(A, f, g, x, x', y, y')$ for which the cardinality of $A$ is not greater than $k$.
The subscript $e$ is sometimes dropped when there is no ambiguity.

\begin{proposition}
$\pbgpcp_{\star}(k)$ and $\pbgpcp_{+}(k)$ are equivalent.
\end{proposition}

\begin{proof}
 Let $\cali = (A, f, g, x, x', y, y')$ be an instance of $\pbgpcp(k)$.
 
First, $\cali$ is a yes-instance of $\pbgpcp_\star$ if, and only if, 
at least one of the following two holds true: 
$x x' = yy'$ or $\cali$ is a yes-instance of $\pbgpcp_+$.
Therefore, $\pbgpcp_\star(k)$ reduces to $\pbgpcp_+(k)$.
Second, $\cali$ is a yes-instance of $\pbgpcp_+$ if, and only if, there exists $a \in A$ such that 
$$
(A, f, g, xf(a),  x', y g(a), y')
$$
is a yes-instance of $\pbgpcp_\star$.
Therefore, $\pbgpcp_+(k)$ reduces to $\pbgpcp_\star(k)$.
\end{proof}

Define $\pbpcp(k)$ as the restriction of $\pbgpcp_+(k)$ to those instances 
$(A, f, g, x, x', y, y')$ that satisfy $xx'yy' = \mv$.
$\pbpcp(\infty)$ is the PCP.
The fundamental property of PCP is its undecidability \cite{Post46PCP,HopcroftMU01,SipserBook,Manna74}.
The undecidabilities of many decision problems are proven by reductions from PCP \cite{HopcroftMU01,Manna74}.
As far as we know, undecidability in $3$-by-$3$ matrices has always been proven by reductions from PCP or GPCP.
Note that the restriction of $\pbpcp(k + 2)$ to \emph{Claus instances}  \cite{HalavaHH07}  is equivalent to $\pbgpcp(k)$ \cite{HalavaHH07,Claus80,HarjuKHandbook}.

Define $\kgp$ as the smallest $k \in \bN$ such that $\pbgpcp(k)$ is undecidable; 
define $\kp$ as the smallest $k \in \bN$ such that $\pbpcp(k)$ is undecidable.
The exact values of $\kp$ and $\kgp$ are still unknown.
However, it is known that 
$\kp \le \kgp + 2$ \cite{HarjuKHandbook}, 
$2 < \kgp$,  \cite{HalavaHH02},
$\kp \le 7$ \cite{MatiyasevichS05},
and 
$\kgp \le 5$ \cite{HalavaHH07}:
$$
3 \le \kgp \le \kp \le \kgp + 2 \le 7 \, .
$$
The decidabilities of $\pbgpcp(3)$, $\pbgpcp(4)$,  $\pbpcp(3)$, $\pbpcp(4)$, $\pbpcp(5)$, and $\pbpcp(6)$  are open.

\subsection{Undecidability bounds} 
\label{sec:33-bounds}

In this section, 
we prove that $\pbZ^3(\kgp)$,  $\pbR^3(\kp)$, and $\pbR^3(\kgp + 1)$ are undecidable;
the undecidabilities of $\pbZ^3(\kgp)$ and $\pbR^3(\kp)$ were already known \cite{Claus80,HalavaHH07}.
However, it is still unknown whether $\kp \le \kgp + 1$.
Besides, the undecidability of $\pbZ^3(\kgp)$ implies
that of  $\pbR^4(\kgp)$ by Proposition~\ref{prop:ZR}
and that of 
$\pbM^3(\kgp + 1)$ by Proposition~\ref{prop:ZM}.
Previous related undecidability results are listed in Table~\ref{tab:33gpcp}.
As $\kgp \le 5$ \cite{HalavaHH07}, 
$\pbZ^3(5)$, $\pbR^3(6)$, $\pbR^4(5)$, and $\pbM^3(6)$ are undecidable.

 \begin{table}
 \begin{center}
 \begin{tabular}{|l|l|l|}
\hline
Year & Undecidable problem & Reference  \\
\hline
1970 & $\pbM^3(2 \kp + 2)$ &  \cite{Paterson70} \\
1974 & $\pbZ^3_{3, 2}(\kp)$ & \cite{Manna74} \\
1980 & $\pbpcp(10)$ &  \cite{Claus80} \\
     & $\pbR^3(\kp)$ & \cite{Claus80} (see also \cite{HarjuKHandbook} and Theorem~\ref{th:pcp-R}) \\
1981 &  $\pbpcp(9)$  & \cite{Pansiot81} (see also \cite{HarjuKK96,HarjuKHandbook}) \\
1996 &  $\pbgpcp(7)$ & \cite{HarjuKK96} (see also \cite{HarjuKHandbook}) \\
1997 & $\pbM^3(2 \kp + 1)$ &  \cite{HarjuKHandbook} \\
1999 & $\pbM^3(\kp + 2)$ & \cite{BournezB99} (see also \cite{BournezB02}) \\    
2001 & $\pbZ_{1, 1}^3(2 \kp)$ & \cite{HalavaH01} \\
     & $\pbM^3(\kp + 1)$ & \cite{HalavaH01} \\
2005 & $\pbpcp(7)$ & \cite{MatiyasevichS05} \\
2007 & $\pbgpcp(5)$ & \cite{HalavaHH07} \\
     & $\pbZ^3(\kgp)$ & \cite{HalavaHH07} (see also Theorem~\ref{th:G-Z}) \\
     & $\pbZ_{1, 1}^3(\kgp + 2)$ & \cite{HalavaHH07} \\
     & $\pbM^3(\kgp + 2)$ & \cite{HalavaHH07} \\
\hline
\end{tabular}
\end{center}
\caption{\label{tab:33gpcp}
Undecidability in $3$-by-$3$ matrices and the (generalized) Post correspondence problem.}
\end{table}

\begin{lemma}[\cite{Manna74,HarjuKHandbook,BellP08bounds,Claus80}] \label{lem:psi-ZRUC}
 There exists a morphism $\Psi\colon \zeon^\star \times \zeon^\star \to \mat{3}{\bQ}$ such that for all $u$, $v \in \zeon^\star$,
 the $(1, 3)$th entry of $\Psi(u, v)$ equals $0$ if, and only if, $u = v$.
\end{lemma}

\begin{proof}
 Let $\sigma\colon \zeontw^\star \to \bN$ be the function defined by:
for each $w \in \zeontw^+$, 
 $w$ is a base-$3$ representation of the integer $\sigma(w)$
(we convene that $\mv$ is a representation of $0$). 
Hence, $\sigma$ satisfies
 $\sigma(\zero) = 0$,
 $\sigma(\one) = 1$,
 $\sigma(\two) = 2$,
 and 
$$
\sigma(w w') = 3^{\lgr{w'}} \sigma(w)   + \sigma(w')
$$
 for all $w$, $w' \in \zeontw^\star$. 
Set 
\begin{equation} \label{eq:def-psi}
 \Psi(u, v) 
 =
 \begin{pmatrix}
 1 & \sigma(v) & \sigma(u) - \sigma(v) \\
 0 & 3^{\lgr{v}} & 3^{\lgr{u}} - 3^{\lgr{v}} \\
 0 & 0 & 3^{\lgr{u}} \
\end{pmatrix} 
\end{equation}
for every $u$, $v \in  \zeon^\star$.
Straightforward computations yield
$$
\Psi(u u', v v')   = \Psi(u, v) \Psi(u', v')  
$$
for all $u$, $v$, $u'$, $v' \in \zeon^\star$, so $\Psi$ is a morphism.
Now, remark that $\sigma$ is not injective because $\sigma(\zero w ) = \sigma(w)$ for every $w \in \zeontw^\star$.
However, $\sigma$ is injective on the set of those words in $\zeontw^\star$ that do not begin with $\zero$.
In particular, $\sigma$ is injective on $\left\{ \one, \two \right\}^\star$.
Since the $(1, 3)$th entry of $\Psi(u, v)$ equals $\sigma(u) - \sigma(v)$ for all $u$, $v \in \zeon^\star$,
$\Psi$ satisfy the desired property.
\end{proof}

Let $\cals$ be a multiplicative semigroup, 
let $A$ and $B$ be alphabets, and 
let $\Psi\colon A^\star \times B^\star \to \cals$  be a morphism.
If the operation of $\cals$ is computable then $\Psi$ is computable.

\begin{theorem}[\cite{HalavaHH07}] \label{th:G-Z} 
 $\pbZ^3(\kgp)$ is undecidable.
\end{theorem}

\begin{proof}
Let us show that $\pbgpcp_+(k)$ reduces to $\pbZ^3(k)$ for any $k$.
Set $E_1   = \begin{pmatrix} 1 &  0 & 0 \end{pmatrix}^ \rmt$ and $E_3  = \begin{pmatrix} 0 & 0 & 1 \end{pmatrix}^ \rmt$; 
such a notation is consistent with Section~\ref{sec:general}.
Let $\Psi$ be as in Lemma~\ref{lem:psi-ZRUC}.

Let $\cali = (A, f, g, x, x', y, y')$ be an instance of $\pbgpcp(k)$.
Put 
 \begin{gather*}
 L  = E_1^\rmt\Psi(x, y) \, , \\
 C  =  \Psi(x', y') E_3 \,, \\
 X(w) =\Psi(f(w), g(w))  
\intertext{for every $w \in A^\star$, and}
  \calx = \left\{ X(a) : a \in A \right\} \, .
 \end{gather*}
Since $\Psi$ is computable, 
$(L, C, \calx)$ is computable from $\cali$.
Moreover, the cardinality of $\calx$ is not greater than that of $A$,
so $(L, C, \calx)$ is an instance of $\pbZ^3(k)$.
To conclude the proof of the theorem, we only need to check that the following three assertions are equivalent:
 \begin{enumerate}
  \item
 $\cali$ is a yes-instance of $\pbgpcp_+$.
  \item
 There exists  $w \in A^+$ such that $L X(w) C = 0$.
  \item
 $(L, C, \calx)$ is a yes-instance of $\pbZ^3$. 
 \end{enumerate}
 For every $w \in A^\star$, the $(1, 3)$th entry of $\Psi(x f(w) x', y g(w) y')$ equals $L X(w) C$, and thus  
\begin{equation} \label{eq:LXC-xfx-ygy}
L X(w) C = 0 \iff x f(w) x' =  y g(w) y' \, .  
\end{equation}
Therefore, the first two considered assertions are equivalent.
Now,  remark that 
$X(w w') = X(w) X(w')$ for all $w$, $w' \in A^\star$.
It follows that 
\begin{equation}  \label{eq:X+-A+}
\calx^+ = \left\{ X(w) : w \in A^+ \right\} \, .  
\end{equation}
Therefore, the last  two considered assertions are equivalent.
\end{proof}
%
%


\begin{theorem}[\cite{Manna74,HarjuKHandbook,Claus80}] \label{th:pcp-R}
$\pbR^3(\kp)$ is undecidable. 
\end{theorem}

\begin{proof}
Let us show that $\pbpcp(k)$ reduces to $\pbR^3(k)$ for any $k$.
Let the notation be as in the proof of Theorem~\ref{th:G-Z}.
Without loss of generality, we assume $\Psi(\mv, \mv) = I_3$.
Hence, if $xx'yy' = \mv$ then $L = E_1^\rmt$ and $C = E_3$.
The following three assertions are thus equivalent in the case where $\cali$ is an instance of $\pbpcp$:
\begin{enumerate}
 \item $\cali$ is a yes-instance of $\pbpcp$.
  \item There exists $w \in A^+$ such that $E_1^\rmt X(w) E_3 = 0$.
  \item $\calx$ is a yes-instance of $\pbR^3$. \qedhere
\end{enumerate}
\end{proof}

The last important result of Section~\ref{sec:three-three} is:

\begin{theorem} \label{th:G-R}
$\pbR^3(\kgp + 1)$ is undecidable.
\end{theorem}

 Our proof of Theorem~\ref{th:G-R}  requires the introduction of additional material, 
 including the proofs of two lemmas.
An instance  $(A, f, g, x, x', y, y')$ of $\pbgpcp$ is called \emph{Claus-like} if it satisfies the following three conditions for any $w \in A^\star$:
   \begin{gather*} 
  x f(w)  \ne y g(w)\,,  \\
    f(w) x'  \ne  g(w)y' \,,
    \intertext{and} 
     f(w) = g(w) \iff w = \mv \, .
   \end{gather*} 
Let  $\lambda$ and  $\rho$ be the morphisms from $\zeon^\star$ to itself defined by:
$\lambda(a) = \one \two a$ and $\rho(a) = a \one \two$ for each $a \in \zeon$.
The useful properties of $\lambda$ and $\rho$ are summarized in the following lemma:

\begin{lemma}
 The following properties hold true for any $u$, $v \in \zeon^\star$: 
 \begin{gather}
  \lambda(u) \one \two  = \one \two  \rho(v) \iff u = v \,,  \label{eq:l12-12r-u-v}\\
  \lambda(u) \ne  \one \two  \rho(v)\,, \label{eq:l-12r}\\
  \lambda(u)  \one \two \ne \rho(v) \,, \label{eq:l12-r}
  \intertext{and} 
  \lambda(u) = \rho(v) \iff uv = \mv  \, .  \label{eq:l-r-u-v}
 \end{gather}
\end{lemma}

\begin{proof}
 The proof of Equation~\eqref{eq:l12-12r-u-v} is left to the reader.
 The length of $\lambda(u)$ is a multiple of $3$ 
 whereas 
 the length of  $\one \two  \rho(v)$  is congruent to $2$ modulo $3$.
  Therefore, Equation~\eqref{eq:l-12r} holds true.
 Equation~\eqref{eq:l12-r} is proven in the same way as Equation~\eqref{eq:l-12r}.
 It remains to prove Equation~\eqref{eq:l-r-u-v}.
 If $uv = \mv$ then $\lambda(u) = \mv = \rho(v)$.
 If $u = \mv$ and $v \ne \mv$ then  $\lambda(u) = \mv \ne \rho(v)$.
 If $u \ne \mv$ and $v = \mv$ then  $\lambda(u) \ne \mv = \rho(v)$.
 Let us now deal with the last case: $u \ne \mv$ and $v \ne \mv$.
 The lengths of $\lambda(u)$ and $\rho(v)$ are then larger than or equal to $3$.
 Furthermore, 
 the second letter of $\lambda(u)$ equals $\one$ 
 whereas 
 the second letter of $\rho(v)$ equals $\two$.
 It follows $\lambda(u) \ne \rho(v)$. 
\end{proof}

\begin{lemma} \label{lem:gpcp-claus}
 For each $e \in \{ {\star}, {+} \}$, 
 $\pbgpcp_e(k)$ reduces to its restriction to Claus-like instances.
\end{lemma}

\begin{proof}
Let $\cali = (A, f, g, x, x', y, y')$ be an instance of $\pbgpcp(k)$.

First, let
$$
\bar A = \left\{ a \in A: f(a) g(a) \ne \mv \right\} \, ,
$$ 
let $\bar f$ be the restriction of $f$ to $\bar A^\star$, 
let $\bar g$ be the restriction of $g$ to $\bar A^\star$, 
and let 
$$
\bar \cali = (\bar A, \bar f, \bar g, x, x', y, y') \, .
$$
It is clear that $\bar \cali$ is an instance $\pbgpcp(k)$ and  that $\bar \cali$ is computable from $\cali$.
Moreover, if $A \ne \bar A$ then 
$\cali$ is a yes-instance of $\pbgpcp_e$ if, and only if, at least one the following two holds true: 
$xx' = yy'$ or $\bar \cali$ is a yes-instance of $\pbgpcp_e$.
Replacing $\cali$ with $\bar \cali$ if needed, 
we may assume that $A = \bar A$, or equivalently, that 
\begin{equation} \label{eq:fw-gw-w}
f(w) g(w) = \mv \iff w = \mv 
\end{equation}
for every $w \in A^\star$.

Now, put 
 \begin{align*}
 \tilde x & = \lambda(x) \,, &   \tilde f & = \lambda \circ f \,,  & \tilde x' & = \lambda(x')\one \two \,,  \\
 \tilde y & = \one \two \rho(y) \,,  & \tilde g & = \rho \circ g \,, &  \tilde y' & =  \rho(y')  \,,
 \end{align*}
and 
$$
\tilde \cali =
(A, \tilde f,  \tilde g, \tilde x, \tilde x', \tilde y, \tilde y') \, .
$$
It is clear that $\tilde \cali$ is an instance of $\pbgpcp(k)$ and that $\tilde \cali$ is computable from $\cali$.
Moreover, let $w \in A^\star$.
By letting $u = x f(w) x'$ and $v = y g(w) y'$ in Equation~\eqref{eq:l12-12r-u-v}, we get 
$$
\tilde x \tilde f(w)  \tilde x'
= 
\tilde y \tilde g(w)  \tilde y'
\iff 
x f(w) x' = y g(w) y' \, .
$$
Therefore, $\cali$ is a yes-instance of $\pbgpcp_e$ if, and only if, $\tilde \cali$ is a yes-instance of $\pbgpcp_e$.
It remains to prove that $\tilde \cali$ is a Claus-like instance of $\pbgpcp$.
By letting $u = x f(w)$ and $v = y g(w)$ in Equation~\eqref{eq:l-12r}, we get 
$$\tilde x \tilde f(w)  \ne  \tilde y \tilde g(w) \, .$$
By letting $u =  f(w)x'$ and $v =  g(w)y'$ in Equation~\eqref{eq:l12-r}, we get 
$$ \tilde f(w) \tilde x' \ne  \tilde g(w)\tilde y \, .$$
By letting $u =  f(w)$ and $v =  g(w)$ in Equation~\eqref{eq:l-r-u-v}, we get 
$$\tilde f(w) = \tilde g(w) \iff f(w) g(w) = \mv \, .$$
Finally, combining the latter equivalence with Equation~\eqref{eq:fw-gw-w} yields 
$$\tilde f(w) = \tilde g(w) \iff w = \mv \, . \qedhere $$
\end{proof}

\begin{proof}[Proof of Theorem~\ref{th:G-R}]
Let us show that $\pbgpcp(k)$ reduces to $\pbR^3(k + 1)$ for any $k$.
By Lemma~\ref{lem:gpcp-claus}, 
we only need to reduce Claus-like instances of $\pbgpcp_\star(k)$.
Let the notation be as in the proof of Theorem~\ref{th:G-Z}.
Without loss of generality, we assume $\Psi(\mv, \mv) = I_3$.
Combining the latter assumption and Equation~\eqref{eq:X+-A+}, we get
\begin{equation} \label{eq:Xstar-Astar}
 \calx^\star = \left\{ X(w) : w \in A^\star \right\} \, .
\end{equation}
To prove the theorem, 
it suffices to check that the following four assertions are equivalent in the case where $\cali$ is a Claus-like instance of $\pbgpcp$:
\begin{enumerate}
 \item $\cali$ is a yes-instance of $\pbgpcp_\star$.
  \item There exists $w \in A^\star$ such that $L X(w) C = 0$.
 \item $0 \in L \calx^\star  C$.
 \item $\calx \cup \{ C L \}$ is a yes-instance of $\pbR^3$. 
\end{enumerate}

The first two considered assertions are equivalent because Equation~\eqref{eq:LXC-xfx-ygy} holds for every $w \in A^\star$.
The second and the third considered assertions are equivalent by Equation~\eqref{eq:Xstar-Astar}.
Let us now show that the last two considered assertions are equivalent.
Let $w \in A^\star$.
Clearly, 
\begin{itemize}
 \item the $(1, 3)$th entry of $\Psi(xf(w), y g(w))$ equals $L X(w)E_3$, 
 \item the $(1, 3)$th entry of $\Psi(f(w) x', g(w) y')$ equals $E_1^\rmt X(w) C$, and 
 \item the $(1, 3)$th entry of $\Psi(f(w), g(w))$  equals $E_1^\rmt X(w) E_3$.
\end{itemize}
As $\cali$ is a Claus-like instance of $\pbgpcp$, 
it follows 
that both $L X(w) E_3$ and $E_1^\rmt X(w)  C$ are non-zero and 
that $E_1^\rmt X(w) E_3 = 0$ is equivalent to $w = \mv$.
Combining the latter facts with Equations~\eqref{eq:Xstar-Astar} and~\eqref{eq:X+-A+}, 
we obtain that $0$ is not in 
$L \calx^\star E_3$, 
$E_1^\rmt \calx^ \star C$, or 
$E_1^\rmt \calx^+ E_3$.
Besides, remark that 
$$
E_1^\rmt \left( \calx \cup \{ CL \}  \right)^+ E_3 
= 
\left( E_1^\rmt \calx^+ E_3 \right)
\cup 
\left( E_1^\rmt \calx^\star C \right)
\left( L \calx^ \star C \right)^ \star 
\left( L  \calx^\star E_3 \right) \, . 
$$
Hence, we have 
$$
0 \in E_1^\rmt \left( \calx \cup \{ CL \}  \right)^+ E_3
\iff 
0 \in L \calx^ \star C \, ,
$$
as desired.
\end{proof}

\section{Trading dimension for matrices} \label{sec:high-dim}

In this section, 
we prove that  $\pbM^{15}(2)$, $\pbZ^5(3)$, and $\pbZ^9(2)$ are undecidable.
Let $d$, $h$, $k \in \bN \setminus \{ 0 \}$.

\begin{theorem} \label{th:M-M}
 $\pbM^d(hk +  1)$ reduces to  $\pbM^{kd}(h +  1)$.
\end{theorem}

\begin{proof}
Let $\calx$ be an instance of  $\pbM^d(hk +  1)$. 
 Write $\calx$ in the form
 $$
 \calx = \{ U \} \cup  \left\{ X_{i, j} : (i, j) \in \seg{1}{h} \times \seg{1}{k} \right\}  \, .
 $$
  Put 
\begin{align*}
 V & = \begin{pmatrix} O & U \\ I_{kd - d} & O \end{pmatrix} \,, &  
 \Gamma & = \begin{pmatrix} I_d \\ O_{kd - d, d}  \end{pmatrix}  \,, &
 Y_i & = \begin{pmatrix} X_{i, 1} &  X_{i, 2} &  X_{i, 3} & \cdots & X_{i, k} \end{pmatrix}  
\end{align*}
for every $i \in \seg{1}{h}$, and
$$
\caly = \{ V \} \cup \left\{ \Gamma Y_i   : i \in \seg{1}{h} \right\}\, .
$$
Clearly, $\caly$ is an instance of $\pbM^{kd}(h + 1)$ and $\caly$ is computable from $\calx$.

 \begin{lemma} \label{lem:PVY-UX}
For every $X_1$, $X_2$, $X_3$, \ldots, $X_k \in \mat{d}{\bQ}$, equality  
 $$
\begin{pmatrix} X_1 & X_2 &  X_3 & \cdots &  X_k  \end{pmatrix} V^\star \Gamma
 = 
\left\{ X_1, X_2, X_3, \dotsc,  X_k \right\}   U^\star
 $$
 holds true.
 \end{lemma}

\begin{proof}
The idea of the proof is simply to compute 
$$
 \begin{pmatrix}  X_1 &  X_2 &  X_3 & \cdots &  X_k  \end{pmatrix} V^n \Gamma
$$
for every $n \in \bN$.
Let us extend the $k$-tuple $\left( X_j \right)_{j \in \seg{1}{k}}$ 
into an infinite sequence $\left( X_j \right)_{j \in \bN \setminus \{ 0 \}}$ of elements of $\mat{d}{\bQ}$ by means of the recurrence formula:
$$X_{j + k} =  X_j U $$
for every $j \in \bN \setminus \{ 0 \}$.
Let $n \in \bN$ and let $j \in \seg{1}{k}$.
Straightforward inductions on $n$ yield
$$
\begin{pmatrix} X_1 & X_2 &  X_3 & \cdots &  X_k  \end{pmatrix} V^n 
= 
\begin{pmatrix} 
X_{n + 1} & X_{n + 2} & X_{n + 3} & \cdots & X_{n + k} 
\end{pmatrix}
$$
and 
$$
X_{j +  k n} =  X_j U^ n \, .
$$
It follows 
$$
\begin{pmatrix} X_1 & X_2 &  X_3 & \cdots &  X_k  \end{pmatrix}  V^{k n + j - 1} \Gamma = X_{j + kn} = X_j U^ n  \,,
$$
which proves the lemma.
\end{proof}

Put 
 $$
\calx' = \left\{ X_{i, j} : (i, j) \in \seg{1}{h} \times \seg{1}{k} \right\} 
$$
and
$$
\caly' = \left\{ Y_i : i \in \seg{1}{h} \right\} \, .
 $$
 
\begin{lemma} \label{lem:PYY-XX}
Equality $\caly' \caly^\star \Gamma = \calx' \calx^ \star$ holds true.
\end{lemma}

\begin{proof}
Lemma~\ref{lem:PVY-UX} ensures
$$
Y_i V^ \star \Gamma
= 
 \left\{ X_{i, 1}, X_{i, 2}, X_{i, 3}, \dotsc,  X_{i, k} \right\} U^ \star 
$$
for every $i \in \seg{1}{h}$, and thus we have 
 \begin{equation} \label{eq:PVY-UX}
\caly' V^ \star \Gamma
= 
\calx'  U^ \star \,. 
\end{equation}
Besides, 
equalities $\calx = \{ U \} \cup \calx' $ and $\caly =  \{ V \} \cup \Gamma \caly'$ yield 
$$
\left(   \calx' U^ \star \right)^+ = \calx' \calx^ \star 
$$
and 
$$
 \left( \caly' V^\star  \Gamma \right)^+ = \caly'  \caly^\star \Gamma \,,
$$
respectively.
Combining the last  two equalities with Equation~\eqref{eq:PVY-UX}, we obtain
$$
\caly' \caly^\star \Gamma  = \left( \caly' V^\star  \Gamma \right)^+ = \left(   \calx'  U^ \star\right)^+ =  \calx' \calx^ \star  \, ,
$$
as desired. 
\end{proof}

Let us now complete the proof of the theorem.
Combining Lemma~\ref{lem:PYY-XX} with inclusions  $\calx' \subseteq \calx$ and $\Gamma \caly' \subseteq \caly$, 
we get
$$
 \caly' \caly^\star \Gamma \subseteq \calx^+
$$
and 
$$
\Gamma \calx'   \calx^\star \caly'   \subseteq \caly^+ \, .
$$
It follows from the former inclusion that 
$O_{kd, kd} \in \caly^ +$ implies $O_{d, d} \in \calx^+$;
the converse follows from the latter inclusion.
Hence,  $\calx$ is a yes-instance of  $\pbM^d$  if, and only if,  $\caly$ is a yes-instance of  $\pbM^{kd}$.
\end{proof}

Since $\pbM^3(\kgp + 1)$ is undecidable (see Section~\ref{sec:three-three}),
it follows from Theorem~\ref{th:M-M} that $\pbM^{3\kgp}(2)$ is undecidable.
As $\kgp \le 5$ \cite{HalavaHH07}, $\pbM^{15}(2)$ is undecidable.

\begin{theorem} \label{th:Z-Z}
 $\pbZ^d(hk + 1)$ reduces to $\pbZ^{kd}(h + 1)$.  
\end{theorem}

\begin{proof}
By Lemma~\ref{lem:S-Star}, 
we only need to reduce non-degenerated instances of~$\pbZ^d(hk + 1)$.
 
 Let $(L, C, \calx)$ be a non-degenerated instance of~$\pbZ^d(hk + 1)$.
Let  $\Gamma$, $U$, $V$, $\caly$, $\calx'$, and $\caly'$  be as in the proof of Theorem~\ref{th:M-M};
additionally, let $\Lambda \in \bQ^{d \times k d}$ be given by:  
$$
\Lambda =  \begin{pmatrix} I_d & I_d & I_d & \cdots & I_d \end{pmatrix} \, .
$$
Put $\cali = (L \Lambda, \Gamma C, \caly)$.
Clearly, 
$\cali$ is an instance of $\pbZ^{kd}(h + 1)$ 
and 
$\cali$ is computable from $(L, C, \calx)$.
To complete the proof of the theorem, 
it suffices to check that 
$(L, C, \calx)$ is a yes-instance of  $\pbZ^d$
if, and only if, 
$\cali$ is a yes-instance of  $\pbZ^{kd}$.

Equalities $\calx = \{ U \} \cup \calx'$ and $\caly = \{ V \} \cup  \Gamma \caly'$ yield
$$
 \calx^ \star = U^ \star \cup  U^\star \calx' \calx^ \star 
 $$
 and 
$$
  \caly^\star = V^\star \cup V^ \star \Gamma \caly' \caly^\star \, ,
 $$ 
respectively.
Moreover, Lemma~\ref{lem:PVY-UX} ensures
$$
\Lambda V^ \star \Gamma  = U^ \star \, .
$$
Hence, using also Lemma~\ref{lem:PYY-XX}, we get
\begin{align*}
\Lambda \caly^\star \Gamma 
& = \Lambda \left(  V^\star \cup V^ \star \Gamma \caly' \caly^\star \right) \Gamma \\
& = (\Lambda V^\star \Gamma) \cup (\Lambda V^ \star \Gamma)( \caly' \caly^\star \Gamma) \\
& = U^ \star \cup U^\star \calx' \calx^ \star   \\
& = \calx^ \star \,,   
\end{align*}
and then
$$
L \Lambda  \caly^\star \Gamma C  =  L \calx^ \star C   \,.
$$
Since $L \Lambda \Gamma C = LC  \ne 0$,  
we obtain
$$
0 \in L\Lambda \caly^+ \Gamma C 
\iff
0 \in L \calx^+ C \,, 
$$
as desired.
\end{proof}


\begin{lemma} \label{lem:L-stab}
Let $\call$ be a non-zero linear subspace of~$\bQ^{1 \times d}$.
Let $\ell$ denote the dimension of~$\call$.
The restriction of $\pbZ^d(k)$ to those instances $(L, C, \calx)$ for which $L \calx^\star \subseteq \call$ 
reduces to~$\pbZ^\ell(k)$.
\end{lemma}

\begin{proof}
First, let us check that there exist
$P \in \bQ^{\ell \times d}$ 
and
$P' \in \bQ^{d \times \ell}$ such that $L P' P = L$ for every $L \in \call$.
Let $P \in \bQ^{\ell \times d}$ be such that the rows of $P$ form a basis of $\call$.
Since the row rank of $P$ is full, 
there exists $P' \in \bQ^{d \times \ell}$ such that  $P P'  = I_\ell$ \cite{Stewart1}.
Hence, we have 
$$\call = \left\{ K P : K \in \bQ^{1 \times \ell} \right\}$$
and 
$K PP'P   = K P$ for every $K \in  \bQ^{1 \times \ell}$.
Therefore, $P'P$ satisfies the desired property.

We are now ready to prove that the considered restriction of $\pbZ^d(k)$ reduces to~$\pbZ^\ell(k)$.

 Let  $(L, C, \calx)$ be an instance of  $\pbZ^d(k)$ such that  $L \calx^\star \subseteq \call$.
 Put 
 $$\cali = (L P' , P C, P \calx  P' ) \, .$$
 Clearly,
 $\cali$ is an instance of $\pbZ^\ell(k)$ and  $\cali$ is computable from  $(L, C, \calx)$.
 Moreover, let $n \in \bN$.
Since  $L \calx^ n P' P  = L \calx^ n$, 
 a straightforward induction on $n$ yields 
 $$
  L P' \left( P \calx P' \right)^n  = L  \calx^ n P' \,,
 $$
 and thus 
 $$
 L P' \left( P \calx  P'  \right)^n P C = L  \calx^n C \, .
 $$
Therefore,  $(L, C, \calx)$ is a yes-instance of  $\pbZ^d$ if, and only if, $\cali$  is a yes-instance of  $\pbZ^\ell$.
\end{proof}


Define $\calm_d$ as the set of those $X \in \mat{d}{\bQ}$ that satisfy the following two equivalent conditions:
\begin{enumerate}
 \item 
 The leftmost column of $X$ equals  $\begin{pmatrix} 1 \\ O_{d - 1,1} \end{pmatrix}$. 
 \item 
 For every $K \in \bQ^{1 \times d}$, the leftmost entry of $K X$ equals the leftmost entry of $K$.
\end{enumerate}
Define $\mathring{\pbZ}^d(k)$ as the restriction of $\pbZ^d(k)$ to those instances $(L, C, \calx)$ for which $\calx \subseteq \calm_d$.

\begin{theorem} \label{th:G-Zo}
 $\mathring{\pbZ}^3(\kgp)$ is undecidable.
\end{theorem}

\begin{proof}
 Let us show that $\pbgpcp_+(k)$ reduces to $\mathring{\pbZ}^3(k)$ for any $k$.
 Let the notation be as in the proof of Theorem~\ref{th:G-Z}.
 By Equation~\eqref{eq:def-psi}, the range of $\Psi$ is a subset of $\calm_3$.
 It follows  that $X(w) \in \calm_3$ for every $w \in A^\star$, 
 and thus  $(L, C, \calx)$ is an instance $\mathring{\pbZ}^3(k)$.
\end{proof}

\begin{lemma} \label{lem:Zo-nondeg}
  $\mathring{\pbZ}^d(k)$ reduces to its restriction to non-degenerated instances.
\end{lemma}

\begin{proof}
The proof is the same as that of Lemma~\ref{lem:S-Star}.
\end{proof}

\begin{theorem} \label{th:Zo-Z}
$\mathring \pbZ^d(hk + 1)$ reduces to $ \pbZ^{1 + k(d - 1)}(h + 1)$.
\end{theorem}

\begin{proof}
The proof relies on Lemma~\ref{lem:L-stab}.
For each $s \in \bQ$, define  $\calk(s)$ as the set of those $K \in \bQ^{1 \times kd}$ such that, 
for every $j \in \seg{0}{k - 1}$,
the $(jd + 1)$th entry of $K$ equals $s$.
Let $K \in \bQ^{1 \times kd}$ and let $K_1$, $K_2$, $K_3$, \ldots, $K_k \in \bQ^{1 \times d}$ be such that
$$
K = \begin{pmatrix} K_1 & K_2 & K_3 & \cdots & K_k \end{pmatrix} \,.
$$
For every $s \in \bQ$, $K$ belongs to $\calk(s)$ if, and only if, the leftmost entry of $K_j$ equals $s$ for every $j \in \seg{1}{k}$.
Put $\call = \bigcup_{s \in \bQ}  \calk(s)$.
Clearly, $\call$ is a linear subspace of $\bQ^{1 \times k d}$ and the dimension of $\call$ equals $1 + k (d - 1)$.
By Lemmas~\ref{lem:L-stab} and \ref{lem:Zo-nondeg}, it suffices to show that 
the restriction of $\mathring \pbZ^d(hk + 1)$ to non-degenerated instances 
reduces to 
the restriction of $ \pbZ^{kd}(h + 1)$ to those instances $(L, C, \calx)$ for which $L \calx^ \star \subseteq \call$.

Let $(L, C, \calx)$ be a non-degenerated instance of $\mathring{\pbZ}^d(hk + 1)$.
Let the notation be as in the proofs of Theorems~\ref{th:M-M} and~\ref{th:Z-Z}.
Let $s$ denote the leftmost entry of $L$.
It is clear that 
$$
L  \Lambda = \begin{pmatrix} L & L & L & \cdots & L \end{pmatrix} \in \calk(s) \, .
$$
Moreover, if $K \in \calk(s)$ then straightforward computations yield  
 $$
K  V  =  \begin{pmatrix} K_2 & K_3 & \cdots & K_k & K_1 U \end{pmatrix} \in \calk(s)
 $$
and 
$$
K \Gamma Y_i  =  K_1 Y_i = \begin{pmatrix} K_1 X_{i, 1} & K_1 X_{i, 2} & K_1 X_{i, 3} & \cdots & K_1 X_{i, k} \end{pmatrix}   \in \calk(s)
$$
for  $i \in \seg{1}{h}$.
Hence, we have $L \Lambda \in \calk(s)$ and $\calk \caly \subseteq \calk(s)$.
It follows $L \Lambda \caly^ \star \subseteq \calk(s) \subseteq \call$, 
and thus $\cali$ is an instance of the suitable restriction of~$\pbZ^{k d}$.
\end{proof}

We claim that $\mathring \pbZ^d(hk + 1)$ reduces to $\mathring{\pbZ}^{1 + k(d - 1)}(h + 1)$; 
the verification is left to the reader.

As $\kgp \le 5$ \cite{HalavaHH07}, $\mathring{\pbZ}^3(5)$ is undecidable by Theorem~\ref{th:G-Zo}.
It then follows from  Theorem~\ref{th:Zo-Z} that $\pbZ^5(3)$ and $\pbZ^9(2)$ are undecidable.
Combining Theorems~\ref{th:G-Zo} and~\ref{th:Zo-Z},   
we obtain that $\pbZ^{2 \kgp - 1}(2)$ is undecidable.
Therefore, $\pbR^{2 \kgp}(2)$  and $\pbM^{2 \kgp - 1}(3)$ are undecidable by Propositions~\ref{prop:ZR} and~\ref{prop:ZM}.

\section{Open questions} \label{sec:open}

The cases where $d = 2$ and where $k = 1$ yield challenging open questions.

\subsection{Two-by-two matrices}

The undecidability of $\pbM^3$ was first proven in 1970 \cite{Paterson70}.
It was later proven that $\pbZ^2(1)$  and $\pbM^2(2)$ are decidable \cite{Vereshchagin85,HHHK-Skolem-05,BournezB02}
and that $\pbM^2$ is NP-hard \cite{BellHP12}.
However, 
the decidabilities of $\pbM^2(k + 1)$, $\pbZ^2(k)$, and $\pbR^2(k)$ remain open for $2 \le k \le \infty$.
The decidability of $\pbM^2$ has been repeatedly reported as open since 1977 \cite{Schultz77}.

By Proposition~\ref{prop:Z2-M2}, $\pbZ^2(2)$ and $\pbM^2(3)$  are equivalent.
By Lemma~\ref{lem:ZijZ}, $\pbR^2(2)$ reduces to $\pbZ^2(2)$ and $\pbM^2(3)$.
Therefore, if there exist $d$, $k \in \bN$ such that $d \ge 2$, $k \ge 2$, $(d, k) \ne (2, 2)$, and $\pbM^d(k)$ is decidable
then $\pbM^3(2)$ or $\pbR^2(2)$ is decidable.
The decidabilities of the latter two problems remain open.

\subsection{Linear recurrences} \label{sec:lin-rec}

The decidability of  $\pbM^d(1)$ is easy to see:
for every $X \in \mat{d}{\bQ}$, 
$\{ X \}$ is a yes-instance  of $\pbM^d(1)$
if, and only if,
$X^d$ equals the $d$-by-$d$ zero matrix.
Moreover, 
it is known that $\pbZ^5(1)$ is decidable \cite{HHHK-Skolem-05}, the proof being highly non-trivial.
However, the decidabilities  of $\pbZ^d(1)$ and $\pbR^d(1)$ remain open for $d \ge 6$.
Let us briefly discuss the question.

Given a sequence  $\left( u_n \right)_{n \in \bN}$  of elements of $\bQ$,
we say that $\left( u_n \right)_{n \in \bN}$ satisfies a \emph{linear recurrence relation (LRR) of order $d$}  if the following three equivalent conditions \cite{SalomaaS78} are met:
\begin{enumerate}
\item 
There exist $L \in \bQ^{1 \times d}$, $C \in \bQ^{d \times 1}$,  and $X \in \mat{d}{\bQ}$ such that 
$u_n = L X^n C$
for every $n \in \bN$. 
\item 
There exist $a_0$, $a_1$, \ldots, $a_{d - 1} \in \bQ$ such that 
$$u_{n + d} = \sum_{i = 0}^{d - 1} a_i u_{n + i} 
$$
for every $n \in \bN$. 
 \item There exists two polynomials $f(\ttx)$ and $g(\ttx)$ over $\bQ$ such that 
 $g(0) \ne 0$, 
 the degree of $f(\ttx)$ is smaller than $d$,
 the degree of $g(\ttx)$ is not greater than $d$,  
 and the generating function of $\left( u_n \right)_{n \in \bN}$ satisfies:
 $$
 \sum_{n = 0}^\infty u_n \ttx^n = \frac{f(\ttx)}{g(\ttx)} \, .
 $$
\end{enumerate}
Let  $\left( u_n \right)_{n \in \bN}$ be a sequence of elements of $\bQ$  that satisfies an LRR of order~$d$.
\begin{itemize}
 \item If $u_0 \ne  0$ then there exist 
$X \in \mat{d}{\bQ}$ such that $u_n u_0^{-1}$ equals the $(1, 1)$th entry of $X^n$ for every $n \in \bN$.
\item If $u_0 = 0$ then there exists 
$X \in \mat{d}{\bQ}$ such that $u_n$ equals the $(1, d)$th entry of $X^n$ for every $n \in \bN$.
\end{itemize}
The following two problems are equivalent to~$\pbZ^d(1)$:
\begin{enumerate}
 \item 
Given $X \in \mat{d}{\bQ}$, decide whether there exists $n \in \bN$ such that the $(1, 1)$th entry of $X^n$  equals~$0$.
 \item 
Given  a sequence $\left( u_n \right)_{n \in \bN}$ of elements of $\bQ$   that satisfies an LRR of order $d$, 
decide whether there exists $n \in \bN$   such that $u_n = 0$.
\end{enumerate}
The following two problems are equivalent to $\pbR^d(1)$:
\begin{enumerate}
 \item 
 Given $X \in \mat{d}{\bQ}$, decide whether there exists $n \in \bN \setminus \{ 0 \}$ such that the $(1, d)$th entry of $X^n$  equals~$0$.
 \item 
 Given a sequence $\left( u_n \right)_{n \in \bN}$ of elements of $\bQ$  that satisfies an LRR of order $d$ and $u_0 = 0$,  
 decide whether there exists $n \in \bN \setminus \{ 0 \}$   such that $u_n = 0$.
\end{enumerate}

\bibliographystyle{plain}
\bibliography{bibmat}

\end{document}